\documentclass[aps,prd,reprint,superscriptaddress,nofootinbib,longbibliography,floatfix]{revtex4-2}

\usepackage[T1]{fontenc}
\usepackage[utf8]{inputenc}
\usepackage{lmodern}
\usepackage{microtype}
\usepackage{amsmath,amssymb,amsthm,mathtools,bm}
\usepackage{booktabs}
\usepackage{graphicx}
\usepackage{enumitem}
\usepackage{aliascnt}
\usepackage{hyperref}
\usepackage[nameinlink,noabbrev]{cleveref}

\hypersetup{
 hidelinks,
 pdftitle={Finite Quantum Histories: Holonomy Spectra, Minimal Clocks, and Exact Clock-Change Covariance},
 pdfauthor={Maxim V. Churilov},
 pdfsubject={Exact finite relational histories and quantum clock transformations},
 pdfkeywords={relational time, quantum clocks, history states, monodromy, holonomy, quantum reference frames}
}
\ifdefined\pdfinfoomitdate\pdfinfoomitdate=1\fi
\ifdefined\pdftrailerid\pdftrailerid{}\fi
\ifdefined\pdfsuppressptexinfo\pdfsuppressptexinfo=-1\fi
\setlist{itemsep=0.15em,topsep=0.35em,parsep=0pt}
\allowdisplaybreaks
\theoremstyle{plain}
\newtheorem{theorem}{Theorem}[section]
\newaliascnt{proposition}{theorem}
\newtheorem{proposition}[proposition]{Proposition}
\aliascntresetthe{proposition}
\newaliascnt{lemma}{theorem}

\aliascntresetthe{lemma}
\newaliascnt{corollary}{theorem}
\newtheorem{corollary}[corollary]{Corollary}
\aliascntresetthe{corollary}
\newaliascnt{rigiditytheorem}{theorem}
\newtheorem{rigiditytheorem}[rigiditytheorem]{Rigidity Theorem}
\aliascntresetthe{rigiditytheorem}

\theoremstyle{definition}
\newaliascnt{definition}{theorem}
\newtheorem{definition}[definition]{Definition}
\aliascntresetthe{definition}
\newaliascnt{construction}{theorem}

\aliascntresetthe{construction}

\theoremstyle{remark}
\newaliascnt{remark}{theorem}
\newtheorem{remark}[remark]{Remark}
\aliascntresetthe{remark}
\newaliascnt{warning}{theorem}
\newtheorem{warning}[warning]{Warning}
\aliascntresetthe{warning}

\crefname{theorem}{Theorem}{Theorems}
\Crefname{theorem}{Theorem}{Theorems}
\crefname{rigiditytheorem}{Rigidity Theorem}{Rigidity Theorems}
\Crefname{rigiditytheorem}{Rigidity Theorem}{Rigidity Theorems}
\crefname{proposition}{Proposition}{Propositions}
\Crefname{proposition}{Proposition}{Propositions}
\crefname{lemma}{Lemma}{Lemmas}
\Crefname{lemma}{Lemma}{Lemmas}
\crefname{corollary}{Corollary}{Corollaries}
\Crefname{corollary}{Corollary}{Corollaries}
\crefname{definition}{Definition}{Definitions}
\Crefname{definition}{Definition}{Definitions}
\crefname{construction}{Construction}{Constructions}
\Crefname{construction}{Construction}{Constructions}
\crefname{remark}{Remark}{Remarks}
\Crefname{remark}{Remark}{Remarks}
\crefname{warning}{Warning}{Warnings}
\Crefname{warning}{Warning}{Warnings}

\newcommand{\C}{\mathbb C}

\newcommand{\Z}{\mathbb Z}
\newcommand{\id}{\mathbf 1}
\newcommand{\Hh}{\mathcal H}

\newcommand{\A}{\mathcal A}
\newcommand{\B}{\mathcal B}
\newcommand{\K}{\mathcal K}

\newcommand{\Tr}{\operatorname{Tr}}
\newcommand{\e}{\mathrm e}
\newcommand{\Fix}{\operatorname{Fix}}
\newcommand{\Ker}{\operatorname{ker}}
\newcommand{\spec}{\operatorname{spec}}

\newcommand{\PU}{\operatorname{PU}}
\newcommand{\norm}[1]{\left\lVert #1\right\rVert}
\newcommand{\abs}[1]{\left\lvert #1\right\rvert}
\newcommand{\ket}[1]{\lvert #1\rangle}
\newcommand{\bra}[1]{\langle #1\rvert}
\newcommand{\ketbra}[2]{\lvert #1\rangle\!\langle #2\rvert}

\newcommand{\Hist}{H_{\mathrm{hist}}}
\newcommand{\Jhist}{\mathsf J}
\newcommand{\ClockAlg}{\mathfrak C}

\makeatletter
\AtBeginDocument{%
  \@ifpackageloaded{hyperref}{\hypersetup{hidelinks}}{}%
}
\makeatother

\begin{document}

\title{Finite Quantum Histories: Holonomy Spectra, Minimal Clocks, and Exact Clock-Change Covariance}
\author{Maxim V. Churilov}
\email{churilovm1305@gmail.com}
\affiliation{Independent Researcher, Orenburg, Russia}
\date{June 22, 2026}

\begin{abstract}
A finite relational history is often written under the restrictive closure
assumption that one unitary step has finite order.  We remove that assumption
and solve the cyclic history problem for an arbitrary time-dependent sequence
of finite-dimensional unitary steps.  The propagation Hamiltonian is a
unitary connection Laplacian on a cycle.  Its complete gauge invariant is the
conjugacy class of the monodromy
$M=U_{L-1}\cdots U_0$, and its full spectrum is
\[
 \lambda_{a,k}=1-\cos\!\left(\frac{2\pi k-\theta_a}{L}\right),
 \qquad e^{i\theta_a}\in\spec M.
\]
Consequently, the exact history sector is isomorphic to $\Fix(M)$, the
frustration is determined by the monodromy eigenphases, and the spectral gap
above a nonempty zero-energy sector is obtained in closed form.  We also derive
an exact Chebyshev determinant and a Bessel--Wilson-loop expansion of the
finite-temperature trace.  The inverse spectral problem is solved as well:
ordinary spectral data recover precisely the multiset of monodromy phase
cosines and are blind to phase orientation.  Low energy
then certifies proximity to an exact relational history and controls all
conditional probabilities with nonvanishing clock weight.

We next define the predictive quotient of a sharp finite clock relative to an
accessible operator system.  It is the unique coarsest event alphabet that
preserves every conditional statistic in a chosen history sector.  We prove a
finite-error recovery theorem: if the minimum diamond separation of
inequivalent event channels exceeds four times the estimation error, simple
threshold clustering recovers the exact quotient, and no approximation with
error below half the separation margin can use fewer records.  For full
matrix access and a homogeneous step $U$, the minimal number of clock events
is the projective order of $U$, not its ordinary order.  We distinguish the
large kinematic normalizer of the clock algebra from the much smaller group
that preserves the coherent history code.  Among transformations that also
preserve the oriented successor relation, exact sharp clock changes are
classified by $U(r)\times\Z_L$ on a rank-$r$ history sector; preserving only
unoriented adjacency replaces $\Z_L$ by the dihedral group.  A reversible channel change
of full-information clock fibers is necessarily unitary, so irreversible
coarse-graining cannot be promoted to exact clock covariance.  Finally,
minimal realizations of a complete history Gram kernel are uniquely unitarily
equivalent, with a quantitative finite-data Procrustes bound.  Independent
finite-matrix code verifies the holonomy reduction, exact spectrum, kernel,
gap, clock compression, and clock-change classification.
\end{abstract}

\maketitle

\section{Introduction}

Relational quantum dynamics replaces an externally prescribed time parameter
by correlations between subsystems of a stationary physical state.  In the
Page--Wootters construction, conditional states relative to clock readings can
obey a Schr\"odinger law even when the total state is invariant under a
constraint \cite{PageWootters1983,Wootters1984}.  This idea has been sharpened
through relative-time observables, constrained quantization, and temporal
quantum reference frames
\cite{LoveridgeMiyadera2019,HoehnVanrietvelde2020,HoehnSmithLock2021,Chataignier2020}.
Quantum-reference-frame transformations make the covariance question
operational: a change of frame is physical only insofar as it preserves the
probabilities assigned to transformed preparations, events, and effects
\cite{Giacomini2019,Vanrietvelde2020,Glowacki2023,Carette2025}.

A complementary line of work encodes a sequence of unitary operations in a
stationary history state.  Feynman--Kitaev propagation Hamiltonians are central
to Hamiltonian complexity and autonomous quantum computation
\cite{Feynman1985,KitaevShenVyalyi2002,CahaLandauNagaj2018,Watson2019}.
Finite clocks can approximate continuous control with rapidly decreasing error
\cite{Woods2019}, while finite-resource temporal measurements can become
nonunitary or time nonlocal unless additional structure is imposed
\cite{SmithAhmadi2019,Hausmann2025}.  History-state Hamiltonians are generically
critical in growing families \cite{GonzalezCubitt2018}; an exact finite theorem
should therefore expose, rather than hide, its clock-length dependence.  A
companion manuscript studies permutation-covariant graph selection and two
special incidence-history Hamiltonians \cite{ChurilovPrelocality2026}.  Here we
remove all graph-specific structure and classify arbitrary finite-dimensional
unitary histories, their predictive clock quotients, and their exact clock-change
symmetries.

Three finite-dimensional questions remain logically distinct.

First, a cyclic history is commonly closed by assuming $U^L=\id$.  This is
unnecessary and obscures the true obstruction.  For time-dependent steps
$U_0,\ldots,U_{L-1}$, the invariant object is their monodromy.  The appropriate
Hamiltonian is a non-Abelian connection Laplacian on the cycle, closely related
to discrete magnetic Laplacians \cite{FabilaCarrasco2018,FabilaCarrasco2022}.
The exact history problem can therefore be solved by holonomy rather than by a
finite-order ansatz.

Second, a clock register can contain redundant labels.  Two readings should be
identified if every accessible conditional experiment gives the same answer at
both readings.  This operational quotient depends on the observable system,
not merely on the cardinality of the auxiliary register.  In a homogeneous
unitary history, the correct invariant is the order of the step in the
projective unitary group.

Third, covariance under simultaneous conjugation is too weak to classify
clock changes.  A generic unitary mixing different event subspaces converts a
sharp clock into an interference measurement.  Conversely, the full
normalizer of a sharp clock algebra contains independent fiber rotations that
do not preserve the coherent history code.  Exact clock covariance requires
both event admissibility and preservation of the physical history sector.

This paper provides one finite theorem package answering these questions.
The main results are:
\begin{enumerate}[label=(\roman*)]
\item a complete gauge classification of cyclic unitary histories by the
conjugacy class of their monodromy;
\item a closed formula for every eigenvalue and eigenvector of the cyclic
history Hamiltonian, including exact kernel and gap formulas;
\item quantitative certification of exact histories from low energy and a
finite-temperature history-weight estimate;
\item the predictive quotient and the projective-order theorem for minimal
sharp clocks;
\item a classification of exact sharp clock transformations preserving the
history code;
\item rigidity of reversible channel-valued clock changes;
\item exact and stable uniqueness of predictively complete minimal history
realizations.
\end{enumerate}

The claims are deliberately finite and exact.  We do not assert that every
physically reasonable clock is sharp, that all internal-time choices are
unitarily equivalent, or that a thermodynamic arrow follows from unitary
history propagation.  The multiple-choice problem can produce inequivalent
quantizations \cite{Malkiewicz2015}; the theorems below identify the precise
finite class in which exact equivalence is forced.

\section{Cyclic unitary histories}
\label{sec:setup}

Let $L\ge2$, let $\Hh_C=\C^L$ with orthonormal basis
$\{\ket t:t\in\Z_L\}$, and let $\Hh_D$ be a finite-dimensional data Hilbert
space of dimension $d$.  A cyclic unitary protocol is a sequence
\begin{equation}
 \boldsymbol U=(U_0,\ldots,U_{L-1}),
 \qquad U_t\in U(\Hh_D).
 \label{eq:unitary-protocol}
\end{equation}
Indices are taken modulo $L$.  Define the propagation differences
\begin{equation}
 A_t=\bra{t+1}\otimes\id-\bra t\otimes U_t
 \label{eq:At}
\end{equation}
and the positive cyclic history Hamiltonian
\begin{equation}
 \Hist(\boldsymbol U)
 =\frac12\sum_{t\in\Z_L}A_t^\dagger A_t.
 \label{eq:history-H}
\end{equation}
For $\ket\Psi=\sum_t\ket t\otimes\ket{\psi_t}$,
\begin{equation}
 \bra\Psi\Hist(\boldsymbol U)\ket\Psi
 =\frac12\sum_t\norm{\psi_{t+1}-U_t\psi_t}^2.
 \label{eq:energy-identity}
\end{equation}
Thus zero energy imposes exact conditional transport.

The ordered partial products and the monodromy are
\begin{align}
 V_0&=\id,
 &V_t&=U_{t-1}\cdots U_0 \quad (1\le t\le L-1),
 \label{eq:partial-products}\\
 M&=U_{L-1}\cdots U_0.
 \label{eq:monodromy}
\end{align}
No finite-order condition is imposed on $M$.

\begin{definition}[Vertex-wise gauge equivalence]
Two protocols $\boldsymbol U$ and $\boldsymbol U'$ are gauge equivalent if
there are unitaries $R_t\in U(\Hh_D)$ such that
\begin{equation}
 U_t'=R_{t+1}U_tR_t^\dagger
 \label{eq:gauge-links}
\end{equation}
for all $t\in\Z_L$.
\end{definition}
The associated block-diagonal unitary
$R=\sum_t\ketbra{t}{t}\otimes R_t$ conjugates the two history Hamiltonians.

\section{Holonomy normal form and exact spectrum}
\label{sec:spectrum}

\begin{theorem}[Holonomy normal form]
\label{thm:normal-form}
Let $M$ be the monodromy \eqref{eq:monodromy} and
\begin{equation}
 G=\sum_{t=0}^{L-1}\ketbra tt\otimes V_t.
 \label{eq:gauge-G}
\end{equation}
Then
\begin{equation}
 G^\dagger\Hist(\boldsymbol U)G
 =\Hist(\id,\ldots,\id,M).
 \label{eq:normal-form}
\end{equation}
Within vertex-wise gauge transformations, the conjugacy class of $M$ is a
complete invariant: two cyclic protocols are gauge equivalent if and only if
their monodromies are unitarily conjugate.
\end{theorem}

\begin{proof}
Write $\psi_t=V_t\phi_t$.  For $0\le t\le L-2$,
\begin{equation}
 \psi_{t+1}-U_t\psi_t
 =V_{t+1}(\phi_{t+1}-\phi_t),
\end{equation}
while the closing difference is
\begin{equation}
 \psi_0-U_{L-1}\psi_{L-1}
 =\phi_0-M\phi_{L-1}.
\end{equation}
Substitution in \eqref{eq:energy-identity} proves
\eqref{eq:normal-form}.

If \eqref{eq:gauge-links} holds, multiplication around the cycle gives
$M'=R_0MR_0^\dagger$.  Conversely, suppose
$M'=R_0MR_0^\dagger$.  Define recursively
$R_{t+1}=U_t'R_tU_t^\dagger$.  The assumed conjugacy of the monodromies makes
the recursion consistent at $t=L-1$, and \eqref{eq:gauge-links} follows.
\end{proof}

The theorem identifies a finite history with a unitary connection on a cycle.
All local link data can be gauged away; the sole obstruction is the holonomy
around the cycle.  Twisted and connection Laplacians are standard spectral
objects \cite{FabilaCarrasco2018,LinWanZhang2024,TorresHugas2026}; the result
below specializes their holonomy structure to a finite history code and then
combines it with operational clock compression and covariance.

\begin{theorem}[Complete spectrum]
\label{thm:complete-spectrum}
Let
\begin{equation}
 M\ket{a}=e^{i\theta_a}\ket a,
 \qquad -\pi<\theta_a\le\pi,
 \label{eq:monodromy-phases}
\end{equation}
where the eigenvectors form an orthonormal basis and phases are repeated with
multiplicity.  For $k=0,\ldots,L-1$, define
\begin{equation}
 q_{a,k}=\frac{2\pi k-\theta_a}{L}.
 \label{eq:quasimomentum}
\end{equation}
Then
\begin{equation}
 \lambda_{a,k}=1-\cos q_{a,k}
 \label{eq:exact-eigenvalues}
\end{equation}
are all eigenvalues of $\Hist(\boldsymbol U)$, with orthonormal eigenvectors
\begin{equation}
 \ket{\Phi_{a,k}}
 =\frac1{\sqrt L}\sum_{t=0}^{L-1}
 e^{iq_{a,k}t}\ket t\otimes V_t\ket a.
 \label{eq:exact-eigenvectors}
\end{equation}
\end{theorem}

\begin{proof}
By \Cref{thm:normal-form}, it is enough to diagonalize the one-link normal
form.  On the monodromy eigenspace generated by $\ket a$, its quadratic form is
\begin{equation}
 \frac12\sum_{t=0}^{L-2}\abs{f_{t+1}-f_t}^2
 +\frac12\abs{f_0-e^{i\theta_a}f_{L-1}}^2.
 \label{eq:twisted-form}
\end{equation}
For $f_t=L^{-1/2}e^{iqt}$, all edges have the same phase difference provided
\begin{equation}
 e^{iqL}=e^{-i\theta_a},
\end{equation}
which is precisely \eqref{eq:quasimomentum}.  The discrete Laplacian then has
eigenvalue $1-\cos q$.  The $L$ Fourier modes are orthonormal for fixed $a$,
and distinct monodromy eigenvectors are orthogonal.  There are $Ld$ vectors,
so the list is complete.  Conjugating back by $G$ yields
\eqref{eq:exact-eigenvectors}.
\end{proof}

\begin{corollary}[Exact history sector]
\label{cor:kernel}
The zero-energy sector is
\begin{equation}
 \Ker\Hist(\boldsymbol U)
 =\left\{
 \frac1{\sqrt L}\sum_{t=0}^{L-1}\ket t\otimes V_t\ket\psi:
 \ket\psi\in\Fix(M)
 \right\}.
 \label{eq:history-kernel}
\end{equation}
In particular,
\begin{equation}
 \dim\Ker\Hist(\boldsymbol U)=\dim\Fix(M).
 \label{eq:kernel-dimension}
\end{equation}
A frustration-free cyclic history exists if and only if the monodromy has a
fixed vector.
\end{corollary}

\begin{proof}
An eigenvalue \eqref{eq:exact-eigenvalues} vanishes exactly when
$q_{a,k}=0$ modulo $2\pi$.  With $-\pi<\theta_a\le\pi$ and
$0\le k<L$, this occurs only for $\theta_a=0$ and $k=0$.
\end{proof}

\begin{corollary}[Exact gap and holonomy frustration]
\label{cor:gap}
Assume $\Fix(M)\ne\{0\}$.  Define
\begin{equation}
 \vartheta(M)=
 \begin{cases}
 2\pi,& M=\id,\\
 \displaystyle\min_{\theta_a\ne0}\abs{\theta_a},&M\ne\id.
 \end{cases}
 \label{eq:vartheta}
\end{equation}
Then the spectral gap above zero is exactly
\begin{equation}
 \Delta(\boldsymbol U)
 =1-\cos\!\left(\frac{\vartheta(M)}{L}\right).
 \label{eq:exact-gap}
\end{equation}
If $\Fix(M)=\{0\}$, the ground energy is
\begin{equation}
 E_0=1-\cos\!\left(
 \frac{\min_a\abs{\theta_a}}{L}
 \right).
 \label{eq:frustrated-ground}
\end{equation}
For $M=\id$,
\begin{equation}
 \Delta=1-\cos\frac{2\pi}{L}
 =\frac{2\pi^2}{L^2}+O(L^{-4}).
 \label{eq:identity-gap}
\end{equation}
\end{corollary}

\begin{proof}
For a nonzero monodromy phase, the smallest absolute representative of
$(2\pi k-\theta_a)/L$ is $\abs{\theta_a}/L$.  Within the fixed subspace,
the first nonconstant clock mode has momentum $2\pi/L$.  Since all nonzero
phases lie in $(0,\pi]$, the stated minimum follows.  The frustrated case is
the same minimization without a zero branch.
\end{proof}

\begin{figure}[t]
 \centering
 \includegraphics[width=0.98\columnwidth]{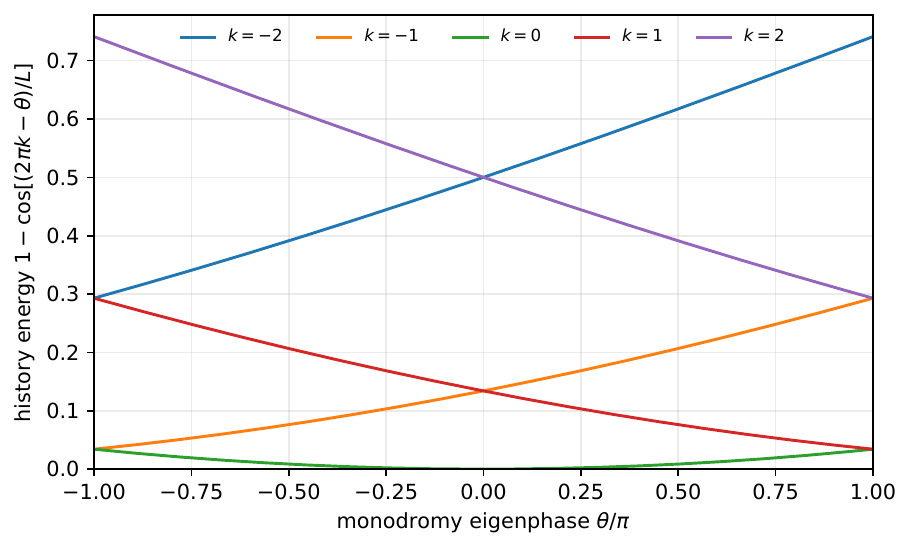}
 \caption{Lowest spectral branches for a cyclic history of length $L=12$.
 A monodromy eigenphase shifts the clock momentum.  Zero energy occurs only at
 a fixed monodromy vector; a small nonzero eigenphase produces a precisely
 controlled frustration energy.}
 \label{fig:spectral-branches}
\end{figure}

\begin{remark}[What the spectrum depends on]
The local distribution of the $U_t$ is gauge.  Spectrum, kernel, and gap depend
only on the eigenphases of the total monodromy.  This is stronger than the
usual statement that a history Hamiltonian can be reduced to a clock
Laplacian when the protocol closes exactly: the theorem also treats partially
closing and fully frustrated protocols without approximation.
\end{remark}

\begin{theorem}[Spectral determinant and thermal Wilson loops]
\label{thm:det-heat}
Let \(T_L\) be the Chebyshev polynomial of the first kind and let
\(I_n\) denote the modified Bessel function.  For every complex \(z\),
\begin{equation}
\begin{split}
 \det\!\left(z\id+\Hist(\boldsymbol U)\right)
 &=2^{d(1-L)}\\
 &\quad\times\det\!\left[
 T_L(1+z)\id-\frac{M+M^\dagger}{2}
 \right].
\end{split}
 \label{eq:cheb-determinant}
\end{equation}
For \(\beta\ge0\), the heat trace is the absolutely convergent series
\begin{equation}
 Z_L(\beta)
 :=\Tr\e^{-\beta\Hist(\boldsymbol U)}
 =Le^{-\beta}\sum_{m\in\mathbb Z}
 I_{mL}(\beta)\Tr(M^{-m}).
 \label{eq:heat-wilson}
\end{equation}
Thus the characteristic polynomial and the full finite-temperature partition
function depend only on gauge-invariant Wilson-loop moments of the monodromy.
\end{theorem}

\begin{proof}
For one monodromy eigenphase \(\theta\), \Cref{thm:complete-spectrum} gives
\begin{equation}
\begin{split}
 &\prod_{k=0}^{L-1}
 \left[z+1-\cos\!\left(\frac{2\pi k-\theta}{L}\right)\right]\\
 &\qquad=2^{1-L}\left[T_L(1+z)-\cos\theta\right].
\end{split}
 \label{eq:cheb-scalar}
\end{equation}
which is the standard factorization of \(T_L(x)-\cos\theta\).  Multiplying
Eq.~\eqref{eq:cheb-scalar} over all eigenphases proves
Eq.~\eqref{eq:cheb-determinant}.

For the heat trace, insert the Jacobi--Anger expansion
\(e^{\beta\cos q}=\sum_{n\in\mathbb Z}I_n(\beta)e^{inq}\) into the exact
spectrum and sum over \(k\).  The clock sum vanishes unless \(n=mL\), in
which case it equals \(L\).  Summing the remaining factor
\(e^{-im\theta_a}\) over \(a\) gives \(\Tr(M^{-m})\).  Absolute convergence
follows from the Bessel expansion of the entire function
\(e^{\beta\cos q}\).
\end{proof}

\begin{theorem}[Spectral inverse and orientation blindness]
\label{thm:spectral-inverse}
The ordinary spectrum of \(\Hist(\boldsymbol U)\) determines exactly the
multiset
\begin{equation}
 \bigl\{\cos\theta_a:a=1,\ldots,d\bigr\},
 \label{eq:cosine-holonomy-data}
\end{equation}
and no finer oriented phase data.  Equivalently, two monodromies give
isospectral history Hamiltonians if and only if their eigenphase cosines agree
with multiplicity.  In particular, independent replacements
\(\theta_a\mapsto-\theta_a\) leave every scalar spectral invariant unchanged.
The heat trace can be written
\begin{equation}
 Z_L(\beta)=Le^{-\beta}\left[
 dI_0(\beta)+2\sum_{m=1}^{\infty}I_{mL}(\beta)
 \operatorname{Re}\operatorname{Tr}(M^m)\right].
 \label{eq:heat-real}
\end{equation}
Thus oriented Wilson moments \(\operatorname{Im}\operatorname{Tr}(M^m)\)
require an additional orientation-sensitive observable; they cannot be
reconstructed from the ordinary spectrum or partition function.
\end{theorem}

\begin{proof}
Equation~\eqref{eq:cheb-determinant} has the form
\(C\,P[T_L(1+z)]\), where
\(P(y)=\prod_a(y-\cos\theta_a)\).  Composition with the nonconstant polynomial
\(T_L(1+z)\) is injective on polynomials, so the characteristic polynomial
uniquely determines \(P\), hence the multiset in
Eq.~\eqref{eq:cosine-holonomy-data}.  Conversely that multiset determines the
characteristic polynomial.  Pairing the \(m\) and \(-m\) terms in
Eq.~\eqref{eq:heat-wilson} and using unitarity gives Eq.~\eqref{eq:heat-real}.
\end{proof}

\section{Conditional dynamics and low-energy certification}
\label{sec:stability}

Let $r=\dim\Fix(M)$ and choose an isometry
$Q:\C^r\to\Hh_D$ with range $\Fix(M)$.  Define
\begin{equation}
 J_t=V_tQ,
 \qquad
 \Jhist=\frac1{\sqrt L}\sum_t\ket t\otimes J_t.
 \label{eq:history-isometry}
\end{equation}
Then $\Jhist^\dagger\Jhist=\id_r$ and
$\operatorname{Ran}\Jhist=\Ker\Hist$.
For a seed density operator $\sigma$ on $\C^r$, the stationary physical state
is
\begin{equation}
 \rho_{\mathrm{hist}}=\Jhist\sigma\Jhist^\dagger.
 \label{eq:mixed-history}
\end{equation}
The clock marginal is exactly uniform, and conditioning on $t$ gives
\begin{equation}
 \rho_{D|t}=J_t\sigma J_t^\dagger.
 \label{eq:conditional-state}
\end{equation}
Thus the global state is stationary under the constraint while its conditional
fibers follow the prescribed unitary protocol.

\begin{theorem}[Energy certification of an exact history]
\label{thm:energy-certification}
Let $P_0$ be the projection onto $\Ker\Hist$, let $\Delta>0$ be the exact gap
\eqref{eq:exact-gap}, and let $\ket\Psi$ be normalized with energy
$E=\bra\Psi\Hist\ket\Psi$.  Then
\begin{align}
 \norm{(\id-P_0)\Psi}^2&\le\frac{E}{\Delta},
 \label{eq:leakage-bound}\\
 \sum_t\norm{\psi_{t+1}-U_t\psi_t}^2&=2E.
 \label{eq:transport-defect}
\end{align}
If $E<\Delta$ and
\begin{equation}
 \ket\Phi=\frac{P_0\ket\Psi}{\norm{P_0\Psi}},
\end{equation}
then
\begin{equation}
 \frac12\norm{
 \ketbra\Psi\Psi-\ketbra\Phi\Phi
 }_1
 \le\sqrt{\frac{E}{\Delta}}.
 \label{eq:trace-history-bound}
\end{equation}
\end{theorem}

\begin{proof}
The spectral theorem gives
$\Hist\ge\Delta(\id-P_0)$, which proves
\eqref{eq:leakage-bound}.  Equation \eqref{eq:transport-defect} is the exact
quadratic-form identity \eqref{eq:energy-identity}.  The trace distance between
pure states is
$\sqrt{1-\abs{\langle\Phi|\Psi\rangle}^2}$, and
$\abs{\langle\Phi|\Psi\rangle}^2=\norm{P_0\Psi}^2$.
\end{proof}

The estimate is sharp for a superposition of one ground vector and one first
excited eigenvector.  It also converts an energy measurement into a uniform
bound on all unconditioned effects.

\begin{corollary}[Conditional-probability stability]
\label{cor:conditional-stability}
Under the hypotheses of \Cref{thm:energy-certification}, set
$\eta=\sqrt{E/\Delta}$.  Let $E_t=\ketbra tt\otimes\id$ and let
$0\le F_t\le E_t$ be any effect localized in clock fiber $t$.  If
$\eta<1/L$, then
\begin{equation}
 \left|
 \frac{\bra\Psi F_t\ket\Psi}{\bra\Psi E_t\ket\Psi}
 -
 \frac{\bra\Phi F_t\ket\Phi}{\bra\Phi E_t\ket\Phi}
 \right|
 \le\frac{2\eta}{L^{-1}-\eta}.
 \label{eq:conditional-bound}
\end{equation}
In particular, if $\eta\le(2L)^{-1}$, the right-hand side is at most
$4L\eta$.
\end{corollary}

\begin{proof}
Trace distance bounds the probability difference of every effect by $\eta$.
For the exact history, $\bra\Phi E_t\ket\Phi=L^{-1}$.  Therefore the actual
denominator is at least $L^{-1}-\eta$.  Writing the difference of the two
ratios and using $\bra\Phi F_t\ket\Phi\le L^{-1}$ gives
\eqref{eq:conditional-bound}.
\end{proof}

\begin{proposition}[Finite-temperature history weight]
\label{prop:gibbs}
Assume the ground energy is zero with multiplicity $r\ge1$.  For the Gibbs
state
\begin{equation}
 \rho_\beta=\frac{e^{-\beta\Hist}}{\Tr e^{-\beta\Hist}},
\end{equation}
\begin{equation}
 1-\Tr(P_0\rho_\beta)
 \le\frac{Ld-r}{r}e^{-\beta\Delta}.
 \label{eq:gibbs-bound}
\end{equation}
\end{proposition}

\begin{proof}
The ground contribution to the partition function is $r$.  Each of the
remaining $Ld-r$ levels contributes at most $e^{-\beta\Delta}$ after the
ground energy is subtracted.  Dividing the excited contribution by the full
partition function and then by its lower bound $r$ proves the claim.
\end{proof}

\section{Predictive quotient and minimal sharp clocks}
\label{sec:minimal-clock}

A clock label is physical only through the conditional distinctions it
supports.  Let $\A\subseteq\B(\Hh_D)$ be an operator system: a linear
self-adjoint subspace containing $\id$.  It represents the effects and
observables accessible to the data laboratory.  The history sector is still
encoded by the isometries $J_t:\C^r\to\Hh_D$ in
\eqref{eq:history-isometry}.

\begin{definition}[Predictive equivalence]
Two clock events $t,s\in\Z_L$ are predictively equivalent relative to
$\A$ if
\begin{equation}
 J_t^\dagger A J_t=J_s^\dagger A J_s
 \qquad\text{for every }A\in\A.
 \label{eq:predictive-equivalence}
\end{equation}
We write $t\sim_\A s$ and call
$\Z_L/{\sim_\A}$ the predictive quotient clock.
\end{definition}

The definition is sector-wide: it compares all seed states, not only one
chosen trajectory.  Indeed, \eqref{eq:predictive-equivalence} is equivalent to
\begin{equation}
 \Tr(J_t\sigma J_t^\dagger A)
 =\Tr(J_s\sigma J_s^\dagger A)
 \label{eq:all-seed-equality}
\end{equation}
for every seed density operator $\sigma$ and every $A\in\A$.

\begin{theorem}[Coarsest sufficient event alphabet]
\label{thm:predictive-quotient}
Let $\pi:\Z_L\to Q$ be any deterministic coarse-graining of clock labels.
All conditional expectations of $\A$ are well defined as functions of
$\pi(t)$ for every seed state if and only if
\begin{equation}
 \pi(t)=\pi(s)\quad\Longrightarrow\quad t\sim_\A s.
 \label{eq:coarse-condition}
\end{equation}
Consequently, the quotient map
$\Z_L\to\Z_L/{\sim_\A}$ is the unique coarsest deterministic clock record that
preserves every conditional statistic in the history sector.  Every
sufficient sharp event alphabet has at least
$\abs{\Z_L/{\sim_\A}}$ events.
\end{theorem}

\begin{proof}
If $\pi(t)=\pi(s)$, a statistic depending only on the coarse label must have
the same value at $t$ and $s$ for all $A$ and all seed states.  By
\eqref{eq:all-seed-equality}, this is exactly
$t\sim_\A s$.  Conversely, if each fiber of $\pi$ is contained in one
predictive-equivalence class, the common conditional functional is independent
of the representative.  The quotient by the maximal equivalence classes is
therefore coarsest, and any other sufficient partition refines it.
\end{proof}

\begin{corollary}[Projective-order theorem]
\label{cor:projective-order}
Suppose $r=d$, $J_t=U^t$, and $\A=\B(\Hh_D)$.  Then
\begin{equation}
 t\sim_\A s
 \quad\Longleftrightarrow\quad
 U^{t-s}\in U(1)\id.
 \label{eq:projective-equivalence}
\end{equation}
If $p$ is the order of the projective class $[U]\in\PU(d)$, the minimal sharp
clock has exactly $p$ events.  A cycle of length $L=mp$ has redundancy factor
$m$ even when the ordinary order of $U$ is larger than $p$.
\end{corollary}

\begin{proof}
Predictive equivalence for the full algebra is equality of conjugation maps:
$U^{-t}AU^t=U^{-s}AU^s$ for every $A$.  Hence $U^{t-s}$ commutes with the full
matrix algebra and is scalar.  The converse is immediate.  The quotient size
is therefore the order of $[U]$ in the projective unitary group.
\end{proof}

\begin{remark}[State-dependent compression]
For a single pure seed $\ket\psi$, fewer events may suffice: one only needs to
separate the projective orbit of $\ket\psi$.  The sector-wide quotient in
\Cref{cor:projective-order} is stronger because it preserves predictions for
all seed states and all data observables.  It is the appropriate notion when
the clock is part of a reusable physical theory rather than a record tailored
to one trajectory.
\end{remark}

\section{Robust recovery of the predictive quotient}
\label{sec:robust-quotient}

The exact quotient is stable whenever inequivalent clock events are separated
by a nonzero operational margin.  Let \(\mathcal C_t\) denote the conditional
channel from the seed system to the accessible output experiment at event
\(t\).  Define
\begin{equation}
 \gamma_\A=
 \min_{t\not\sim_\A s}
 \norm{\mathcal C_t-\mathcal C_s}_\diamond,
 \label{eq:predictive-margin}
\end{equation}
with the convention \(\gamma_\A=+\infty\) if there is only one quotient class.

\begin{theorem}[Finite-error quotient recovery]
\label{thm:robust-quotient}
Suppose estimates \(\widehat{\mathcal C}_t\) satisfy
\begin{equation}
 \norm{\widehat{\mathcal C}_t-\mathcal C_t}_\diamond\le\varepsilon
 \qquad(t\in\Z_L).
 \label{eq:channel-estimation-error}
\end{equation}
If \(\gamma_\A>4\varepsilon\), then every threshold
\begin{equation}
 2\varepsilon\le\tau<\gamma_\A-2\varepsilon
 \label{eq:quotient-threshold}
\end{equation}
recovers the exact predictive quotient from the rule
\[
 t\widehat\sim s
 \quad\Longleftrightarrow\quad
 \norm{\widehat{\mathcal C}_t-
 \widehat{\mathcal C}_s}_\diamond\le\tau.
\]
In particular, \(\widehat\sim\) is already an equivalence relation and equals
\(\sim_\A\).
\end{theorem}

\begin{proof}
If \(t\sim_\A s\), then \(\mathcal C_t=\mathcal C_s\), so the estimated
distance is at most \(2\varepsilon\).  If \(t\not\sim_\A s\), the reverse
triangle inequality gives an estimated distance at least
\(\gamma_\A-2\varepsilon\).  The interval
\eqref{eq:quotient-threshold} therefore separates the two cases exactly.
\end{proof}

\begin{theorem}[Robust minimal-clock obstruction]
\label{thm:robust-clock-lower-bound}
Let a possibly approximate clock record \(q(t)\) assign to each event a
representative channel \(\mathcal D_{q(t)}\), and suppose
\begin{equation}
 \max_t\norm{\mathcal C_t-\mathcal D_{q(t)}}_\diamond\le\delta.
 \label{eq:approximate-clock-error}
\end{equation}
If \(2\delta<\gamma_\A\), then \(q\) cannot merge two distinct predictive
classes.  Hence every \(\delta\)-accurate clock has at least
\(\abs{\Z_L/{\sim_\A}}\) records.
\end{theorem}

\begin{proof}
If \(q(t)=q(s)\), the triangle inequality gives
\[
 \norm{\mathcal C_t-\mathcal C_s}_\diamond\le2\delta.
\]
This contradicts the definition of \(\gamma_\A\) whenever \(t\) and \(s\)
are inequivalent.
\end{proof}

The margin \(\gamma_\A\) separates a structural resource from a statistical
one.  Quotient size determines the exact event alphabet; the ratio
\(\gamma_\A/\varepsilon\) determines whether that alphabet is identifiable
from finite-precision channel data.

\section{Kinematic and physical clock changes}
\label{sec:clock-changes}

The exact history code occupies only an $r$-dimensional subspace of the
$Lr$-dimensional direct sum of its conditional fibers.  This distinction is
what reduces the clock-change group.

Define the fiber isometries
\begin{equation}
 \widetilde J_t:\C^r\to\Hh_C\otimes\Hh_D,
 \qquad
 \widetilde J_t\ket\psi=\ket t\otimes J_t\ket\psi,
 \label{eq:fiber-isometries}
\end{equation}
and the history support
\begin{equation}
 \widehat\K=\bigoplus_{t\in\Z_L}\widetilde J_t\C^r.
 \label{eq:history-support}
\end{equation}
On $\widehat\K$, let
\begin{equation}
 E_t=\widetilde J_t\widetilde J_t^\dagger,
 \qquad
 \ClockAlg=\operatorname{span}\{E_t:t\in\Z_L\}.
 \label{eq:clock-algebra}
\end{equation}
The coherent history code is
$\K_0=\operatorname{Ran}\Jhist\subset\widehat\K$.

\begin{proposition}[Kinematic normalizer]
\label{prop:kinematic-normalizer}
A unitary $W$ on $\widehat\K$ normalizes the sharp clock algebra,
$W\ClockAlg W^\dagger=\ClockAlg$, if and only if
\begin{equation}
 W=\sum_{t\in\Z_L}
 \widetilde J_{\sigma(t)}G_t\widetilde J_t^\dagger,
 \qquad
 \sigma\in S_L,
 \quad G_t\in U(r).
 \label{eq:block-monomial}
\end{equation}
Thus the kinematic normalizer is the wreath product
$U(r)^L\rtimes S_L$.  If adjacency on the unoriented cycle must be preserved,
$S_L$ is reduced to the dihedral group $D_L$; preserving orientation reduces
it to the cyclic subgroup $\Z_L$.
\end{proposition}

\begin{proof}
The $E_t$ are the minimal projections of the finite abelian algebra
$\ClockAlg$.  Any algebra normalizer permutes them, giving $\sigma$.  The
restriction of $W$ from the $r$-dimensional range of $E_t$ to that of
$E_{\sigma(t)}$ is a unitary $G_t$, yielding \eqref{eq:block-monomial}.  The
converse is immediate.  The automorphism group of a cycle graph is $D_L$, and
its orientation-preserving subgroup is $\Z_L$.
\end{proof}

The independent $G_t$ are kinematically allowed but generally destroy the
coherent relation among fibers.  Exact relational covariance requires
preservation of $\K_0$.

\begin{rigiditytheorem}[Exact sharp clock changes]
\label{thm:sharp-clock-rigidity}
Let $W$ normalize $\ClockAlg$.  Then
\begin{equation}
 W\K_0=\K_0
 \label{eq:preserve-code}
\end{equation}
if and only if there exist one permutation $\sigma\in S_L$ and one unitary
$R\in U(r)$ such that
\begin{equation}
 W=W_{\sigma,R}:=
 \sum_{t\in\Z_L}
 \widetilde J_{\sigma(t)}R\widetilde J_t^\dagger.
 \label{eq:physical-clock-change}
\end{equation}
Equivalently,
\begin{equation}
 W_{\sigma,R}\Jhist=\Jhist R.
 \label{eq:history-intertwining}
\end{equation}
Consequently, without additional clock-geometry constraints the exact
code-preserving sharp clock-change group is
\begin{equation}
 U(r)\times S_L.
 \label{eq:unordered-group}
\end{equation}
If the oriented successor relation $t\mapsto t+1$ must be preserved, the group
reduces to
\begin{equation}
 U(r)\times\Z_L,
 \label{eq:oriented-group}
\end{equation}
and preserving only unoriented cycle adjacency gives $U(r)\times D_L$.
\end{rigiditytheorem}

\begin{proof}
Every $W$ satisfying \eqref{eq:preserve-code} induces a unique unitary
$R=\Jhist^\dagger W\Jhist$ on the seed space.  Hence
$W\Jhist=\Jhist R$.  Write $W$ in the block-monomial form
\eqref{eq:block-monomial}.  Comparing the component in the
$\sigma(t)$ fiber on both sides gives
\begin{equation}
 G_t=R
\end{equation}
when expressed in the canonical fiber coordinates supplied by
$\widetilde J_t$; in ambient notation this is exactly
\eqref{eq:physical-clock-change}.  Conversely, summing
\eqref{eq:physical-clock-change} over $t$ proves
\eqref{eq:history-intertwining}.

Composition obeys
$W_{\sigma,R}W_{\tau,S}=W_{\sigma\circ\tau,RS}$, which proves
\eqref{eq:unordered-group}.  Restriction to oriented or unoriented cycle
automorphisms gives \eqref{eq:oriented-group} and its dihedral counterpart.
\end{proof}

\begin{corollary}[Covariance of conditional probabilities]
\label{cor:clock-covariance}
Let $\rho=\Jhist\varrho\Jhist^\dagger$ be a physical history state and let
$W=W_{\sigma,R}$ be an exact sharp clock change.  The transformed conditional
state in event $\sigma(t)$ is
\begin{equation}
 J_{\sigma(t)}R\varrho R^\dagger J_{\sigma(t)}^\dagger.
 \label{eq:transformed-conditional}
\end{equation}
If an effect is simultaneously transported by the corresponding fiber
unitary, every conditional probability is unchanged.  The clock change acts
only by event relabeling and one seed-space frame transformation.
\end{corollary}

\begin{proof}
Equation \eqref{eq:history-intertwining} gives the transformed global state.
Projection onto the $\sigma(t)$ fiber gives
\eqref{eq:transformed-conditional}.  Invariance of the Born rule under
simultaneous unitary conjugation proves the last statement.
\end{proof}

\begin{warning}[Arbitrary basis rotation is not a sharp clock change]
A unitary that coherently mixes different $E_t$ without normalizing
$\ClockAlg$ defines a different POVM with interference between event labels.
It may be physically meaningful, but it is not a re-expression of the same
sharp conditional question.  Conversely, a normalizer with independent
$G_t$ is a valid kinematic relabeling of fibers but is not an exact symmetry of
the history code unless the $G_t$ are transported copies of one $R$ as in
\eqref{eq:physical-clock-change}.
\end{warning}

\section{Predictively complete histories}
\label{sec:gram}

The sharp-clock classification assumes a fixed finite event algebra.  A more
general comparison can be made directly from all intervention histories.
Let $\mathcal W$ be a set of finite words labeling preparations,
interventions, and outcomes.

\begin{definition}[History realization]
A history realization is a Hilbert space $\K$ and a map
$\xi:\mathcal W\to\K$.  Its history Gram kernel is
\begin{equation}
 K(w,w')=\langle\xi(w')|\xi(w)\rangle.
 \label{eq:gram-kernel}
\end{equation}
The realization is minimal if
$\operatorname{span}\{\xi(w):w\in\mathcal W\}=\K$.
\end{definition}

The kernel contains not only classical probabilities but relative phases and
all interference data encoded by the chosen history vectors.  The next result
is the finite-dimensional uniqueness theorem for minimal Kolmogorov
realizations of a positive kernel \cite{Aronszajn1950}.

\begin{theorem}[Unitary uniqueness of complete minimal histories]
\label{thm:gram-uniqueness}
Let $(\K_1,\xi_1)$ and $(\K_2,\xi_2)$ be minimal finite-dimensional
realizations of the same Gram kernel:
\begin{equation}
 \langle\xi_1(w')|\xi_1(w)\rangle
 =\langle\xi_2(w')|\xi_2(w)\rangle
 \label{eq:same-gram}
\end{equation}
for all $w,w'$.  Then there is a unique unitary
$U_{21}:\K_1\to\K_2$ satisfying
\begin{equation}
 U_{21}\xi_1(w)=\xi_2(w)
 \label{eq:unitary-realizations}
\end{equation}
for every $w$.  If appending a symbol $a$ is represented by linear operators
$T_i(a)\xi_i(w)=\xi_i(wa)$, then
\begin{equation}
 U_{21}T_1(a)=T_2(a)U_{21}.
 \label{eq:word-intertwining}
\end{equation}
For three minimal realizations, the transition unitaries obey the exact cocycle
law $U_{31}=U_{32}U_{21}$.
\end{theorem}

\begin{proof}
On finite linear combinations define
\begin{equation}
 U_{21}^{(0)}\sum_w c_w\xi_1(w)
 =\sum_w c_w\xi_2(w).
\end{equation}
If the left combination vanishes, the squared norm of the right combination
is the same quadratic form in the common Gram kernel and also vanishes.  The
map is therefore well defined and preserves inner products.  Minimality makes
its domain and range the full spaces, so it extends to a unitary.  Uniqueness
holds on the spanning set.  Equation \eqref{eq:word-intertwining} follows on
generators, and the cocycle law follows from uniqueness.
\end{proof}

\begin{proposition}[Finite-data stability]
\label{prop:gram-stability}
Let $w_1,\ldots,w_N$ span two finite history realizations, and let
$G_i\in M_N(\C)$ be their Gram matrices.  After adjoining zero directions if
necessary, there exists a partial isometry $U$ such that
\begin{equation}
 \sum_{n=1}^N
 \norm{U\xi_1(w_n)-\xi_2(w_n)}^2
 \le\norm{G_1-G_2}_1.
 \label{eq:procrustes-bound}
\end{equation}
\end{proposition}

\begin{proof}
Let $X_i$ be the matrices whose columns are the history vectors, so
$X_i^\dagger X_i=G_i$.  The unitary Procrustes minimum is the squared Bures
distance of the Gram matrices,
\begin{align}
 \min_U\norm{UX_1-X_2}_2^2
 &=\Tr G_1+\Tr G_2 \notag\\
 &\quad-2\Tr\!\left(G_1^{1/2}G_2G_1^{1/2}\right)^{1/2}.
\end{align}
The Powers--St\o{}rmer inequality bounds this quantity by
$\norm{G_1-G_2}_1$ \cite{Bhatia1997}.  The Frobenius norm is the sum in
\eqref{eq:procrustes-bound}.
\end{proof}

The exact theorem says that two finite clocks carrying the same complete
history kernel are not different theories: after redundant directions are
removed, they are one unitary realization.  The stability theorem also makes
the limitation explicit.  Incomplete or ill-conditioned tomography need not
produce a well-controlled clock transformation even when observed kernels are
close.

\section{Channel rigidity and the boundary of covariance}
\label{sec:channel-rigidity}

For unsharp or coarse clocks, a change of temporal description may naturally
be represented by channels rather than unitaries.  Exact covariance, however,
requires reversible preservation of all states in both directions.

\begin{rigiditytheorem}[Reversible quantum channels are unitary]
\label{thm:channel-rigidity}
Let
\begin{equation}
 \Gamma:M_d(\C)\to M_e(\C),
 \qquad
 \Lambda:M_e(\C)\to M_d(\C)
\end{equation}
be completely positive trace-preserving maps satisfying
\begin{equation}
 \Lambda\circ\Gamma=\id_{M_d},
 \qquad
 \Gamma\circ\Lambda=\id_{M_e}.
 \label{eq:mutual-channel-inverse}
\end{equation}
Then $d=e$ and there is a unitary $U\in U(d)$ such that
\begin{equation}
 \Gamma(\rho)=U\rho U^\dagger,
 \qquad
 \Lambda(\rho)=U^\dagger\rho U.
 \label{eq:unitary-channel}
\end{equation}
\end{rigiditytheorem}

\begin{proof}
The two maps are affine inverses between the full state spaces.  Their real
affine dimensions are $d^2-1$ and $e^2-1$, hence $d=e$.  Let $\{A_i\}$ and
$\{B_j\}$ be Kraus families for $\Gamma$ and $\Lambda$.  The composition
$\Lambda\circ\Gamma$ has Kraus operators $B_jA_i$.  The identity channel has
Choi rank one, so every $B_jA_i$ is proportional to $\id$.  At least one
product, say $B_{j_0}A_{i_0}$, is nonzero and therefore invertible.  For every
$i$, $B_{j_0}A_i$ is scalar; multiplying by $B_{j_0}^{-1}$ shows that all
$A_i$ are proportional to $A_{i_0}$.  Thus $\Gamma$ has Kraus rank one:
$\Gamma(\rho)=A\rho A^\dagger$.  Trace preservation gives
$A^\dagger A=\id$, and because $A$ is square, it is unitary.  The formula for
$\Lambda$ follows from \eqref{eq:mutual-channel-inverse}.  This is the standard
reversibility rigidity of finite quantum channels
\cite{NielsenChuang2010}.
\end{proof}

\begin{corollary}[No exact covariance by irreversible clock compression]
\label{cor:no-coarse-covariance}
Suppose two clock fibers retain the full matrix state space and are declared
exactly equivalent by mutually inverse CPTP transformations.  Then the fiber
change is unitary.  Any genuinely irreversible coarse-graining, dephasing, or
dimension reduction can define at most a one-way simulation or an approximate
clock relation, not exact covariance.
\end{corollary}

This result separates two notions often conflated under ``change of clock.''
A reversible exact change belongs to the unitary class already classified in
\Cref{thm:sharp-clock-rigidity}.  A nonunitary relation may still be useful,
but its lost distinctions must be recorded as a physical approximation.

\section{Examples}
\label{sec:examples}

\subsection{A partially closing qubit history}

Let $L\ge2$ and take a monodromy
\begin{equation}
 M=\ketbra00+e^{i\theta}\ketbra11,
 \qquad 0<\abs\theta\le\pi.
 \label{eq:qubit-monodromy}
\end{equation}
Any sequence of links with this product is gauge equivalent.  The history
Hamiltonian has one exact zero mode,
\begin{equation}
 \ket{\Psi_0}
 =\frac1{\sqrt L}\sum_t\ket t\otimes V_t\ket0,
\end{equation}
and a second low branch with ground energy
\begin{equation}
 1-\cos(\theta/L).
\end{equation}
The exact gap above zero is
\begin{equation}
 \min\left\{1-\cos(2\pi/L),\,1-\cos(\abs\theta/L)\right\}
 =1-\cos(\abs\theta/L).
\end{equation}
Thus an arbitrarily small monodromy phase creates an arbitrarily soft
orthogonal history sector without changing any local link norm.

\subsection{A clock longer than its predictive period}

Let
\begin{equation}
 U=\operatorname{diag}(1,\omega,\omega^2),
 \qquad \omega=e^{2\pi i/5},
\end{equation}
and choose $L=15$.  Although the auxiliary clock has 15 labels, the projective
order is five.  The full conditional channel repeats after five steps, so the
predictive quotient has five events and redundancy factor three.  Retaining
all 15 labels can still be useful as a computational record, but it is not a
minimal physical clock for the data algebra.

\subsection{Gauge distribution does not alter the spectrum}

Fix any target monodromy $M$ and choose arbitrary unitaries
$U_0,\ldots,U_{L-2}$.  Setting
\begin{equation}
 U_{L-1}=M(U_{L-2}\cdots U_0)^\dagger
\end{equation}
produces the same exact spectrum for every choice of the first $L-1$ links.
The links can be individually far apart in operator norm, yet the history
spectra coincide because they represent the same cycle holonomy.

\section{Verification and falsification protocol}
\label{sec:verification}

The source package includes \texttt{verify\_finite\_histories.py}.  The
verification method is calibrated before it is applied to the new formulas:
for identity links it reproduces the textbook spectrum of the cycle Laplacian
with maximum absolute error $2.45\times10^{-15}$.

The independent tests then perform the following checks.
\begin{enumerate}[label=(\roman*)]
\item Random Haar protocols with $(L,d)=(3,2),(4,3),(7,2),(8,3)$ are
diagonalized directly.  Their full spectra agree with
\eqref{eq:exact-eigenvalues} to maximum error
$1.34\times10^{-15}$.
\item The block gauge transformation \eqref{eq:gauge-G} agrees with the
one-link monodromy normal form in operator norm to
$3.93\times10^{-15}$.
\item The Chebyshev determinant identity \eqref{eq:cheb-determinant}
and the orientation-blind inverse classification in
\Cref{thm:spectral-inverse} are tested on random protocols and complex spectral parameters with maximum
absolute error $5.34\times10^{-15}$.  The Bessel--Wilson heat trace
\eqref{eq:heat-wilson} agrees with direct diagonalization to
$1.78\times10^{-15}$.
\item A four-dimensional target monodromy with a two-dimensional fixed space
has numerical kernel dimension two.  For $L=9$, the numerical and predicted
gaps agree to $3.0\times10^{-16}$.
\item The history map \eqref{eq:history-isometry} is isometric, lies in the
kernel, and has uniform clock marginal to errors below
$8.4\times10^{-16}$.
\item A ground--first-excited superposition saturates the leakage estimate
\eqref{eq:leakage-bound} to numerical precision.
\item The projective-order example above compresses a 15-event cycle to five
events.
\item A classified clock automorphism with an arbitrary random event
permutation satisfies $W_{\sigma,R}\Jhist=\Jhist R$ to
$5.5\times10^{-16}$.  Independent random fiber gauges, although kinematically
allowed, leak from the coherent code with operator norm $0.977$.
\item A random finite-data instance satisfies the Procrustes--Gram stability
bound \eqref{eq:procrustes-bound}; the squared alignment error is $0.462$
against a trace-norm upper bound $9.263$.
\end{enumerate}
The machine-readable results are included as
\texttt{verification\_results.json}.  These computations test algebraic
implementation and finite instances; they do not replace the proofs.

The most direct falsification tests for the paper are equally explicit:
\begin{enumerate}[label=(\alph*)]
\item a cyclic protocol whose spectrum is not determined by its monodromy
eigenphases would refute \Cref{thm:complete-spectrum};
\item a sufficient deterministic clock quotient smaller than
$\Z_L/{\sim_\A}$ would refute \Cref{thm:predictive-quotient};
\item a clock-algebra normalizer preserving the history code but not of the
form \eqref{eq:physical-clock-change} would refute
\Cref{thm:sharp-clock-rigidity};
\item mutually inverse nonunitary CPTP maps on full matrix algebras would
refute \Cref{thm:channel-rigidity}.
\end{enumerate}

\section{Discussion}
\label{sec:discussion}

The cyclic history Hamiltonian has an exact geometric interpretation.  It is a
connection Laplacian on a one-dimensional closed complex, and the monodromy is
its Wilson loop.  This viewpoint yields more than a reformulation.  It removes
the finite-order assumption, gives the full spectrum in one expression, and
turns failure of history closure into a quantitative holonomy frustration.
The $L^{-2}$ clock scaling is explicit, while an anomalously small monodromy
phase can create a still softer branch.

The predictive quotient makes clock minimality operational.  A clock is not
minimal because its register has the smallest obvious dimension; it is minimal
when no two events are indistinguishable by all accessible conditional
experiments.  For complete access to a homogeneous unitary sector, the relevant
period is projective.  Global phases do not create new physical time events.
With restricted observables, the quotient can be smaller and is determined by
the compressed Heisenberg actions $J_t^\dagger\A J_t$.

The clock-change classification separates three layers:
\begin{center}
\small
\begin{tabular}{@{}ll@{}}
\toprule
Layer & Exact transformation group \\
\midrule
Event algebra only & $U(r)^L\rtimes S_L$ \\
Coherent unordered history & $U(r)\times S_L$ \\
Coherent unoriented cycle & $U(r)\times D_L$ \\
Coherent oriented cycle & $U(r)\times\Z_L$ \\
\bottomrule
\end{tabular}
\end{center}
The collapse from independent fiber gauges to one transported seed unitary is
the content of physical history preservation.  It rules out the claim that an
arbitrary rotation of the clock basis is merely a change of temporal frame.

The Gram-kernel theorem covers a broader situation: two auxiliary clock
constructions may have different coordinates and different redundant
registers, yet represent the same complete intervention history.  After
minimalization, equality of the complete kernel forces a unique unitary
intertwiner.  Partial data give only a stability statement controlled by the
Gram-matrix error and by the conditioning of the chosen tomography.

Several boundaries remain.  First, unsharp POVM clocks require a separate
analysis; they need not admit a sharp event algebra or unitary conditional
evolution.  Second, a one-way channel between clock descriptions may be useful
without being an equivalence.  Third, no thermodynamic arrow is derived here.
The exact propagation law is reversible, and entropy production requires a
macroalgebra, a coarse-graining, and a dynamical relaxation theorem.  Fourth,
the finite cyclic construction does not by itself supply a continuum time
observable.  Finite quasi-ideal clocks address that approximation problem by a
different mechanism \cite{Woods2019}.

Within these boundaries, the cyclic unitary history problem studied here is
closed: arbitrary unitary links reduce to monodromy, the spectrum and exact
history sector are known, minimal event content is computable, and exact sharp
clock covariance is classified rather than postulated.

\section{Conclusion}

A finite cyclic quantum history is controlled globally by one unitary:
its monodromy.  The associated history Hamiltonian is gauge equivalent to a
single twisted edge, and its complete spectrum is the set of shifted cycle
Laplacian branches \eqref{eq:exact-eigenvalues}.  Exact histories are precisely
monodromy-fixed vectors; the gap and frustration are exact functions of the
monodromy phases.

Operational minimality then removes redundant clock events, with projective
order replacing ordinary unitary order under full data access.  Exact sharp
clock changes are not arbitrary basis rotations.  They are the block-monomial
transformations that preserve the coherent history code, giving
$U(r)\times\Z_L$ in the oriented case.  Reversible channel covariance adds no
larger class because mutually inverse finite quantum channels are unitary.
Complete minimal history realizations are likewise unique up to one unitary,
with a quantitative finite-data stability bound.

These results provide an exact finite foundation on which approximate,
unsharp, continuum, and many-body clock limits can be studied without
confusing kinematic relabeling, predictive compression, and physical
covariance.

\begin{acknowledgments}
The author thanks the developers of NumPy, SciPy, Matplotlib, and the
open-source \LaTeX{} ecosystem used for the reproducibility checks and
manuscript production.
\end{acknowledgments}

\bibliographystyle{apsrev4-2}
\bibliography{references}

\end{document}